\documentclass[sigconf,nonacm]{aamas}
\usepackage{amsmath,amsthm}
\usepackage{booktabs,tabularx}
\usepackage{balance}
\usepackage{tikz}
\usetikzlibrary{positioning,arrows.meta}

\setcopyright{none}
\acmConference[Preprint]{Preprint}{October 2026}{}{}
\submissionType{Public preprint}

\title[Organising Trajectory Evidence for Agent Assurance]{Organising Trajectory Evidence for Language-Model Agent Assurance: Fragments, Methods, and the Residual}
\author{Xiaowei Huang}
\affiliation{\institution{University of Liverpool}\city{Liverpool}\country{United Kingdom}}
\email{xiaowei.huang@liverpool.ac.uk}
\hypersetup{pdfauthor={Xiaowei Huang},pdfsubject={Language-model agent assurance}}

\newtheorem{definition}{Definition}
\newtheorem{proposition}{Proposition}
\newcommand{\Runs}{\mathit{Runs}}
\newcommand{\Log}{\mathit{Log}}
\newcommand{\App}{\mathit{App}}
\newcommand{\Viol}{\mathit{Viol}}
\newcommand{\Exec}{\mathit{Exec}}
\newcommand{\Prop}{\mathit{Prop}}
\newcommand{\Past}{\mathbf{O}}
\newcommand{\Always}{\mathbf{G}}
\newcommand{\Eventually}{\mathbf{F}}
\newcommand{\Sig}{\Sigma}
\newcommand{\Dist}{\mathcal D}
\newcommand{\Req}{\mathcal R}
\newcommand{\Lobs}{\mathcal L_{\mathrm{obs}}}
\newcommand{\Lloc}{\mathcal L_{\mathrm{loc}}}
\newcommand{\Lact}{\mathcal L_{\mathrm{act}}}
\newcommand{\Lrule}{\mathcal L_{\mathrm{rule}}}
\newcommand{\Gdep}{\mathbf{G}_{\mathrm{dep}}}

\begin{abstract}
Methods for assessing language-model agents include rule checkers over logs, analyses of skill coverage and composition, support checkers, prefix monitors, execution gates, and rare-event estimators. Each observes a different part of a run and makes a claim of a different strength, and no common account says how these claims combine or what they leave unchecked. We give one, built from a logic, information fragments, and an assurance ledger. Requirements are formulas of a two-tier logic: an outer finite-trace temporal logic over recorded events, and an inner logic of standing over the argument structure that the agent's recorded context supports at a decision. A rule violation and an action taken on withdrawn support are thus formulas of one language. The information available to an assessor, to the agent, and to an execution gate defines fragments of that language; each checking method decides one fragment and returns a typed claim: exact, on a named test suite, a risk bound, or descriptive. The ledger merges the evidence and classifies every obligation as established, addressed but not established, or unaddressed. On 456 released $\tau^2$-bench telecom trajectories, the benchmark oracle flags 231 runs; adding a rule checker, a support stand-in, and a prefix-monitor baseline raises the union to 391, 401, and 407, each contributing flags the others miss, and a gate makes one prohibition exact on a gated deployment. The 49 unflagged runs and the unmet or unaddressed obligations form the residual. Discovery makes unknown requirements explicit, and new or stronger checkers then reduce it.
\end{abstract}
\ccsdesc[500]{Computing methodologies~Intelligent agents}
\ccsdesc[300]{Software and its engineering~Software verification and validation}
\keywords{Language-model agents, agent verification, argumentation, runtime monitoring, assurance cases, policy compliance}

\begin{document}
\maketitle

\section{Introduction}

Benchmarks usually judge a language-model agent by whether it reaches the expected outcome. Assurance asks a harder question: whether the whole course of action was acceptable, including the rules followed, the information used, and the reasons for each decision. Evidence for these questions now comes from several methods, but their results cannot simply be accumulated. They concern different failures, observe different parts of a run, and support conclusions of different strength. A useful assurance framework must therefore say which method can decide each requirement, how far its conclusion extends, and what remains unchecked.

Existing work provides the pieces: rule and support checkers, composition analyses, prefix monitors, execution gates, rare-event estimators, and outcome oracles~\cite{skillcov,skillfuzz,gc,prefixguard,scarce}. Each is useful, but no common account states how their claims combine or where their joint boundary lies. This is an agent-assurance problem: the properties concern what an agent knows locally, which actions it may take, how it interacts with its environment, and what evidence justifies a claim about a deployment. Section~\ref{sec:related} compares existing frameworks; here the methods expose the missing structure.

The need is sharpest because agent failures hide at different depths. Consider four, each surviving the natural fix for the previous one. In a released telecom run, the agent reads that a line's contract has ended, restores the line anyway, and is scored as a success. The failure hides in the \emph{score}, which checks the outcome rather than the path. Checking the whole run against the rules catches this case, but rules see actions rather than what the agent knew. An agent can act on a fact it has seen corrected, so the failure may instead hide in the \emph{facts it holds}~\cite{gc}. Tracking which facts still stand catches that case. Yet an agent with correct, current facts can still act against them. In another run, it reads twice that a customer has data left, says so, and then sells a top-up whose only justification is exhausted data. Here the failure hides in the \emph{decision}. Checking each action against its reasons catches this case, but every such check can still pass when true facts have been selected to mislead~\cite{lyingtruths}. That failure hides in \emph{where the facts come from}; our framework does not decide it, and Section~\ref{sec:residual} uses it to show where requirements must be discovered.

Three difficulties stand in the way of organising the methods that address these failures. First, they have no common language. A rule violation, an ungrounded action, an alarm, and an oracle failure are stated in different terms. Second, their claims differ in strength. A gate's invariant, a count on a test suite, and a monitor's recall are not interchangeable; sources that share labels are one witness rather than two~\cite{bloomfield2003,littlewood2007}. Third, no method states what it cannot see. We resolve all three (Figure~\ref{fig:overview}).
\begin{enumerate}
\item \emph{A two-tier logic over trajectories} (Section~\ref{sec:model}): an outer finite-trace temporal logic over recorded events, and an inner logic of standing over the argument structure extracted from the agent's context. Policy rules, the support discipline, the gate invariant, and the outcome and warning conditions are formulas of one language, and the inner tier names failures of reasoning that no formula over events can state.
\item \emph{Three viewpoints as fragments} (Sections~\ref{sec:fragments} and~\ref{sec:methods}). The information available to the offline assessor, the agent at a decision, and an online gate defines fragments of the logic, and each checking method decides one fragment with a stated claim strength. Rules entered as standing arguments show the gate fragment to be the part of the agent fragment where standing is decided by a rule. A ledger merges the evidence against acceptance contracts and reports each obligation as established, addressed but not established, or unaddressed.
\item \emph{An incremental evaluation} on 456 released $\tau^2$-bench telecom trajectories (Section~\ref{sec:study}), adding methods one at a time to answer the questions below.
\item \emph{A stress test and an account of the residual} (Sections~\ref{sec:stress} and~\ref{sec:residual}): six further method families are placed, two of them beyond single-run properties; and we state what the methods leave unchecked, how discovery makes missing requirements explicit, and how stronger evidence reduces the residual.
\end{enumerate}
The evaluation asks three questions. \textbf{RQ1}: does each added method find failures missed earlier? \textbf{RQ2}: do the methods measure different things, or repeat one another? \textbf{RQ3}: what does the merged evidence establish, and what remains? In brief, starting from the 231 runs the benchmark oracle fails, a rule checker, a support stand-in and a prefix-monitor baseline raise the flagged union to 391, 401 and 407, each contributing runs that no earlier method flags; the monitor's overlap with the oracle (161 runs against 93 under independence) reflects shared labels, not independent confirmation; and an execution gate turns one prohibition into an exact guarantee on a gated deployment, while counts on the released suite establish other obligations only for that suite. The 49 unflagged runs are not evidence of compliance. With the unmet and unaddressed obligations, they form the residual, which the framework names explicitly.

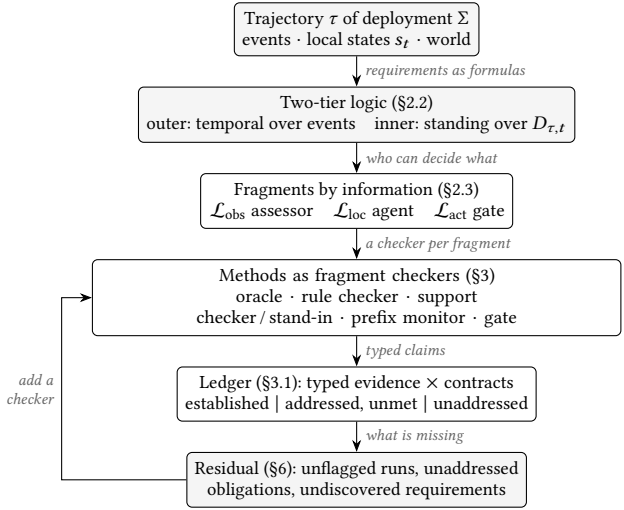
\begin{figure}[t]
\centering
\begin{tikzpicture}[
  font=\footnotesize,
  box/.style={draw,rounded corners=2pt,align=center,inner sep=3pt,minimum height=6mm},
  lab/.style={font=\scriptsize\itshape,text=black!60},
  >={Stealth[length=1.6mm]}
]
\node[box,fill=black!4] (traj) {Trajectory $\tau$ of deployment $\Sig$\\ events $\cdot$ local states $s_t$ $\cdot$ world};
\node[box,below=4mm of traj,fill=black!4] (logic) {Two-tier logic (\S\ref{sec:logic})\\ outer: temporal over events\quad inner: standing over $D_{\tau,t}$};
\node[box,below=4mm of logic] (frag) {Fragments by information (\S\ref{sec:fragments})\\ $\Lobs$ assessor\quad$\Lloc$ agent\quad$\Lact$ gate};
\node[box,below=4mm of frag,text width=0.8\columnwidth] (src) {Methods as fragment checkers (\S\ref{sec:methods})\\ oracle $\cdot$ rule checker $\cdot$ support checker\,/\,stand-in $\cdot$ prefix monitor $\cdot$ gate};
\node[box,below=4mm of src] (ledger) {Ledger (\S\ref{sec:ledger}): typed evidence $\times$ contracts\\ established $\mid$ addressed, unmet $\mid$ unaddressed};
\node[box,below=4mm of ledger,fill=black!4] (res) {Residual (\S\ref{sec:residual}): unflagged runs, unaddressed\\ obligations, undiscovered requirements};
\draw[->] (traj) -- node[lab,right] {requirements as formulas} (logic);
\draw[->] (logic) -- node[lab,right] {who can decide what} (frag);
\draw[->] (frag) -- node[lab,right] {a checker per fragment} (src);
\draw[->] (src) -- node[lab,right] {typed claims} (ledger);
\draw[->] (ledger) -- node[lab,right] {what is missing} (res);
\coordinate (L) at ([xshift=-4mm]src.west);
\draw[->] (res.west) -- (L |- res.west) -- node[lab,left,align=right] {add a\\ checker} (L |- src.west) -- (src.west);
\end{tikzpicture}
\caption{The framework. Requirements are formulas; each method decides a fragment with a stated strength; the ledger merges typed claims and names the residual; discovery makes requirements explicit, and checkers reduce it.}
\Description{A vertical flow diagram: trajectories of a deployment feed the two-tier logic, which is split into three information fragments; methods check fragments and send typed claims to the ledger, which outputs the residual; an arrow from the residual back to the methods is labelled add a checker.}
\label{fig:overview}
\end{figure}

\section{A two-tier logic over trajectories}
\label{sec:model}

\subsection{Deployments and trajectories}

The system under assessment is a deployment $\Sig=(M,P,S,T,\Pi,E)$, comprising the language model and inference configuration $M$, prompt $P$, policy and procedure documents $S$, tools $T$, permissions $\Pi$, and environment $E$. The environment includes the user or user simulator and task distribution. Changing any component may change the runs. $\Runs(\Sig)$ is the set of finite completed runs, and $\Dist_\Sig$ is a specified distribution over them. Universal claims range over $\Runs(\Sig)$; statistical claims name $\Dist_\Sig$, which a released corpus does not by itself specify.

At position $t$ of a run $\tau$ the agent has local state $s_t$, the text it has produced and received so far, and proposes an action $a_t$, a message or a tool call. The assessor receives $\Log(\tau)$: the recorded interaction plus any explicitly available task records, which may include benchmark fault metadata never revealed to the agent. A proposed call and an executed call are distinguished; they differ once a gate is introduced.

This notation separates three viewpoints that are easy to conflate. The agent chooses from what is present in $s_t$; an online gate sees $s_t$ and the proposed action before execution; and an offline assessor sees the completed log, possibly with extra metadata. A method can make an exact judgement only when its required information is available from its viewpoint. Section~\ref{sec:fragments} formalises this.

\subsection{The two tiers}
\label{sec:logic}

The \emph{outer tier} is finite-trace temporal logic over recorded events, with $\Always$ (always), $\Eventually$ (eventually) and $\Past$ (at some position up to and including the present) under the usual finite-trace semantics~\cite{lichtenstein1985}. Data variables bind identifiers, such as a mobile line, and quantify over the finite identifiers recorded in the run. Its atoms are \emph{event} predicates over messages, calls and results, such as $\Prop(a)$ and $\Exec(a)$, that action $a$ is proposed or executed, and $e(\ell)$, that a returned record has shown an ended contract for line $\ell$; and \emph{world} predicates over the environment's true state, such as $e^\star(\ell)$, that the contract has actually ended.

The \emph{inner tier} speaks about what the agent's recorded context represents as currently standing when it acts; it makes no claim about the model's internal beliefs. Its models are the symbolic structures that Grounded Continuation maintains over a conversation~\cite{gc}.

\begin{definition}[Argument structure]
\label{def:argstruct}
An \emph{argument structure} is $D=(V,\mathit{Att},\mathit{Dep},\lambda)$: a finite set $V$ of \emph{arguments}, each a premise taken from an observation or a claim that an action should be taken; an \emph{attack} relation $\mathit{Att}\subseteq V\times V$, read ``$x$ contradicts $y$''; a \emph{dependency} relation $\mathit{Dep}\subseteq V\times V$, read ``$x$ rests on $y$'', required to be acyclic; and a \emph{standing} labelling $\lambda:V\to\{\text{active},\text{weakened},\text{abandoned},\text{resolved}\}$. Following Grounded Continuation, the set of arguments in good standing is $\mathit{St}=\lambda^{-1}(\{\text{active},\text{resolved}\})$, and $D$ is well formed iff $\mathit{St}$ is conflict-free: no two standing arguments attack each other.
\end{definition}

In plain terms, $D$ records what currently supports an action, what contradicts that support, and which earlier premises the support depends on; the standing set contains the arguments that remain usable after corrections and retractions. In the refuel run of Section~\ref{sec:each}, the readings at steps 20 and 46 are \emph{active} premises, ``data is exhausted'' is attacked by them and so does not stand, and the argument for \texttt{refuel\_data} rests on it.

The structure at position $t$ is extracted from the agent's local state by an interpreter, $D_{\tau,t}=\mathit{ext}(s_t)$, together with a partial map $\mathit{arg}_{\tau,t}$ from proposed actions to the argument that the action should be taken next, where one exists. The interpreter builds $D_{\tau,t+1}$ from $D_{\tau,t}$ by classifying each new utterance or tool result as one of the eight operations of Grounded Continuation~\cite{gc}, which commit content, change standing or add attacks without erasing provenance, or pose queries; an undermined or revised argument leaves $\mathit{St}$ together with every argument resting on it through $\mathit{Dep}^{*}$, and nothing is deleted. Definition~\ref{def:argstruct} keeps only the argumentation and dependency components of Grounded Continuation's dependency map, which also carries a plausibility model and public commitments. Grounded Continuation proves that $\mathit{St}$ is conflict-free after every operation and that retraction removes from $\mathit{St}$ exactly the arguments resting on the withdrawn premise~\cite{gc}. These are properties of the structure; whether it describes the agent is a property of $\mathit{ext}$, listed as an assumption wherever used.

Inner-tier formulas are evaluated at an argument of a structure:
\[
\chi\ ::=\ \mathit{standing}\ \mid\ \neg\chi\ \mid\ \chi\wedge\chi\ \mid\ \Gdep\chi,
\]
with $D,v\models\mathit{standing}$ iff $v\in\mathit{St}$, the Boolean cases as usual, and $D,v\models\Gdep\chi$ iff $D,v'\models\chi$ for every $v'$ with $v\,\mathit{Dep}^{*}\,v'$, including $v$ itself. The \emph{bridge} $\langle\mathit{arg}(a)\rangle\chi$ holds at $(\tau,t)$ iff $\mathit{arg}_{\tau,t}(a)$ is defined and $D_{\tau,t},\mathit{arg}_{\tau,t}(a)\models\chi$. The formula that matters is
\begin{equation}
\mathit{grounded}(a)\;\equiv\;\langle\mathit{arg}(a)\rangle\,\Gdep\,\mathit{standing},
\label{eq:grounded}
\end{equation}
that $a$ has an argument, and that argument and everything it transitively rests on still stand. If $a$ has no argument, $\mathit{grounded}(a)$ is false. Because $\mathit{Dep}$ is finite and acyclic, $\mathit{grounded}(a)$ is decided by one traversal of the structure, with no model call.

The two tiers meet by lifting: an inner-tier formula is a step predicate of the outer tier, evaluated at $(\tau,t)$ on $D_{\tau,t}$, and the outer temporal operators range over positions as usual. The inner tier has no notion of time and the outer none of a premise; ``the agent acted without standing support'' needs both: $\Eventually\,\exists a\,(\Exec(a)\wedge\neg\mathit{grounded}(a))$. This one formula covers the second and third failures of the introduction: acting on a fact that has been corrected, and acting against facts that still stand.

The distinction between the tiers is operational. The outer tier can establish that a prohibited action occurred after a recorded event, but it cannot by itself distinguish a deliberate violation from a decision based on stale support. The inner tier can expose the stale support, but it does not say whether the resulting action violates policy. Combining them lets the ledger record these as different obligations rather than collapse them into one failure count.

\paragraph{Requirement families.}
A policy rule $c$ has applicability $\rho_c$ and obligation $\chi_c$, both outer-tier:
\begin{equation}
\varphi_c=\Always\,\forall x\,(\rho_c(x)\Rightarrow\chi_c(x)),\quad
\App_c=\Eventually\,\exists x\,\rho_c(x),\quad
\Viol_c=\neg\varphi_c.
\label{eq:rule}
\end{equation}
The running example is
\begin{equation}
\varphi_7=\Always\,\forall\ell\,\bigl(e(\ell)\Rightarrow\Always\neg\Exec(\mathtt{resume}(\ell))\bigr),
\label{eq:expired}
\end{equation}
over the \emph{recorded} contract state; the property ``never restore a line whose contract has actually ended'' replaces $e$ by $e^\star$, and the two coincide only under assumptions about record freshness and availability. The other families are: the support ideal $\Always\,\forall a\,(\mathit{checked}(a)\Rightarrow\mathit{grounded}(a))$, that every action presented to the support checker rests on standing premises; the gate invariant for a rejection predicate $g$ on the gated deployment $\Sig_g$,
\begin{equation}
\forall\tau\in\Runs(\Sig_g):\ \Always\,\forall a\,\bigl(\Exec(a)\Rightarrow g(s,a)=0\bigr);
\label{eq:gate}
\end{equation}
the outcome condition $\neg F$, with $F$ the oracle's failure event; and the warning ideal, that every failing run alarms at least $h$ steps before its end. Statistical requirements bound the probability of violating one of these formulas (Section~\ref{sec:ledger}). A failure that no formula expresses is outside the framework by construction; Section~\ref{sec:residual} treats that boundary as the place where requirements must be discovered.

\subsection{Fragments by information}
\label{sec:fragments}

An observation map $q$ induces an equivalence $z\sim_q z'$ iff $q(z)=q(z')$, on complete runs or on decision positions $(\tau,t)$. A predicate $f$ is \emph{$q$-determined} if it is constant on every class; this is the indistinguishability account of information in multi-agent systems~\cite{fagin1995}.

\begin{proposition}[Information needed for an exact verdict]
\label{prop:info}
There is a function $h$ with $f=h\circ q$ iff $f$ is $q$-determined. Hence if two inputs have the same observation and opposite truth values, no checker using only that observation is correct on both.
\end{proposition}
\begin{proof}
If $f=h\circ q$, equal observations give equal values. Conversely assign to each observation the common value of $f$ on its inverse image, well defined by invariance. This is an existence claim, not a computability claim.
\end{proof}

Three observation maps define three fragments. The \emph{assessor fragment} $\Lobs$ contains the formulas determined by $\Log(\tau)$ on complete runs. The \emph{agent fragment} $\Lloc$ contains the formulas $\Always\psi$ whose step predicate $\psi$ at $(\tau,t)$ is determined by $(s_t,a_t)$, the agent's own information at the decision. The \emph{gate fragment} $\Lact$ contains the formulas of the guard shape $\Always\,\forall a\,(\Prop(a)\Rightarrow\neg g(s,a))$ with $g$ rejecting only tool calls. Write $\Lrule$ for the rules of the form~\eqref{eq:rule} over event predicates.

\begin{proposition}[Fragment structure]
\label{prop:fragments}
(i) $\Lact\subsetneq\Lloc\subseteq\Lobs$ whenever $\Log$ permits reconstruction of every $(s_t,a_t)$. (ii) $\Lrule\subseteq\Lobs$. (iii) $\Lrule$ and $\Lloc$ are incomparable: $\varphi_7\in\Lrule\setminus\Lloc$ and the support ideal is in $\Lloc\setminus\Lrule$. (iv) If, for each of the three interfaces, some two admissible runs agree on its observation but differ on a world predicate, then a formula over that predicate alone, such as the $e^\star$ form of $\varphi_7$, is in none of the fragments.
\end{proposition}
\begin{proof}
(i) A guard predicate is a function of $(s_t,a_t)$, so $\Lact\subseteq\Lloc$; the support ideal is in $\Lloc$, since $\mathit{checked}(a_t)$ and $D_{\tau,t}=\mathit{ext}(s_t)$ are functions of $(s_t,a_t)$, but not in $\Lact$, since $\mathit{checked}$ includes messages and no guard rejects a message. Under reconstruction every step predicate of $\Lloc$ is log-determined, so $\Lloc\subseteq\Lobs$. (ii) Event predicates and their temporal combinations are log-determined on complete runs. (iii) The instance of $\varphi_7$ at $t$ depends on positions after $t$, which $(s_t,a_t)$ does not fix, since two runs can agree at $t$ and diverge afterwards; the support ideal contains an inner-tier formula and so is not of the shape~\eqref{eq:rule} over event predicates. (iv) By the hypothesis and Proposition~\ref{prop:info}, the predicate is not determined by the log, by $(s_t,a_t)$, or by a guard. The hypothesis fails only if the environment records its complete true state, which a released corpus does not.
\end{proof}

$\Lobs$ and $\Lloc$ are defined by information alone; $\Lact$ is restricted by both information and enforceability, since a gate can only refuse calls. Membership is semantic: for example, $\varphi_7\notin\Lloc$, yet on completed runs it is equivalent to
\begin{equation}
\Always\,\forall\ell\,\bigl(\Exec(\mathtt{resume}(\ell))\Rightarrow\neg\Past e(\ell)\bigr),
\label{eq:past}
\end{equation}
and the proposed-call form of~\eqref{eq:past}, with $\Prop$ for $\Exec$, is in $\Lact$ provided $\Past e(\ell)$ is retained in $s_t$. A gate can enforce the recorded prohibition; ``perform a diagnostic after every repair'' has no equivalent in $\Lact$, since refusing calls cannot make an action happen.

\section{Methods as fragment checkers}
\label{sec:methods}

Each checking method decides one target formula in one fragment and returns a claim of one type; discovery methods instead enlarge the requirement set. Table~\ref{tab:methods} places the methods used in this paper; Section~\ref{sec:study} adds them to a corpus one at a time.

\begin{table*}[t]
\caption{Methods as fragment checkers. The claim type is what the method returns; the blind spot is the class of failures outside its target formula, not a defect of its implementation. Discovery methods enlarge the requirement set rather than check it.}
\label{tab:methods}
\small
\begin{tabularx}{\textwidth}{@{}>{\raggedright\arraybackslash}p{2.7cm}>{\raggedright\arraybackslash}p{2.1cm}>{\raggedright\arraybackslash}p{3.6cm}>{\raggedright\arraybackslash}p{2.6cm}>{\raggedright\arraybackslash}X@{}}
\toprule
Method & Reads & Target formula & Claim type & Cannot see \\
\midrule
Outcome oracle & final state, required actions & $\neg F$ & suite count & how the outcome was reached \\
Rule checker (Skill Coverage~\cite{skillcov}) & $\Log(\tau)$, complete run & $\varphi_c\in\Lobs$ for each audited rule & suite count $(n_{\mathrm{app}},n_{\mathrm{viol}})$ & unlisted requirements; the agent's premises; world predicates \\
Support checker (Grounded Continuation~\cite{gc}) & $(s_t,a_t)$, the agent's prefix & $\Always\forall a(\mathit{checked}(a)\Rightarrow\mathit{grounded}(a))\in\Lloc$ & exact for the extracted structure; suite count & whether a grounded action is correct or compliant \\
Prefix monitor (PrefixGuard~\cite{prefixguard}) & prefix $\tau_{\le t}$ & warning of the event its labels name & recall at a false-alarm rate & any event absent from its labels \\
Boundary gate & $(s_t,a_t)$ at one call & guard form of a rule in $\Lact$ & exact invariant on $\Sig_g$ & obligations to act; the run after a refusal \\
Rare-event estimator (SCARCE~\cite{scarce}) & re-sampled runs & $\Pr(\neg\varphi)$ for a score-defined event & lower risk bound & events its score does not order \\
Composition analysis (SkillFuzz~\cite{skillfuzz}) & skill descriptions & new candidate requirements & discovery & single-run reasoning \\
\bottomrule
\end{tabularx}
\end{table*}

\paragraph{Checkers by fragment.}
The \emph{oracle} targets $\neg F$, a formula of $\Lobs$ over the final position; it is existential where policies are often universal, since it can require that an action occurred but never that one did not. The \emph{rule checker} evaluates $\App_c$ and $\Viol_c$ on $\Log(\tau)$ for each audited rule and returns $(N_{\mathcal T}(\App_c),N_{\mathcal T}(\Viol_c))$; Skill Coverage~\cite{skillcov} extracts rules from a document as when--shall pairs with source spans, and its applicability counts measure how often a rule's situation was exercised. The \emph{support checker} maintains the structure of Definition~\ref{def:argstruct} per turn and decides $\mathit{grounded}(a)$ before $a$ executes, in time linear in the structure~\cite{gc}; it holds the agent to its own information rather than the assessor's log, and is exact for the extracted structure while incomplete for compliance, since a grounded action can violate a rule and an ungrounded one can turn out right. The \emph{prefix monitor} trains a scorer on prefixes against labels, excluding the last $h$ steps before a failure~\cite{prefixguard}; its claim is a recall at a false-alarm rate for whatever event its labels name. The \emph{gate} executes a call only if $g=0$, so~\eqref{eq:gate} holds on $\Sig_g$ by the execution rule, provided every relevant call is routed through it and the checked state is the one used at execution; it is the narrowest checker, since obligations to act have no guard form. A \emph{rare-event estimator} splits a rare event into nested stages along a learned score and multiplies conditional estimates, reaching below $1/N$~\cite{scarce}. Composition analysis finds requirements~\cite{skillfuzz}.

\paragraph{Rules as standing arguments.}
The tiers also compose. Enter a rule $c$ whose past form is in $\Lact$ as a standing argument $A_c\in V$ with $(A_c,\mathit{arg}(a))\in\mathit{Att}$ for any proposed action $a$ violating $\chi_c$ once $\rho_c$ holds in the retained history. We assume \emph{policy priority}: $A_c$ stays standing while its recorded premises hold, and an argument it attacks is labelled non-standing. Conflict-freeness alone does not decide which of the two leaves $\mathit{St}$.

\begin{proposition}[Support subsumes enforcement on $\Lact$]
\label{prop:subsume}
With $A_c$ entered under policy priority and $\mathit{ext}$ faithful to the retained history, for every proposed call $a$, $g_c(s_t,a)=1$ implies $\neg\mathit{grounded}(a)$ at $(\tau,t)$, where $g_c$ is the guard of $c$. Conversely, if $\mathit{arg}(a)$ is defined, every argument it transitively rests on other than itself stands, and $A_c$ is the only standing attacker of $\mathit{arg}(a)$, which would be active without standing policy attackers, then $\neg\mathit{grounded}(a)$ implies $g_c(s_t,a)=1$. Hence an execution rule that rejects ungrounded calls establishes~\eqref{eq:gate} for $g_c$ on the gated deployment, refusing in addition only calls ungrounded for another reason; the refusal names the attacking argument and is retracted by withdrawing $A_c$.
\end{proposition}
\begin{proof}
$g_c(s_t,a)=1$ iff $\rho_c$ holds in the history and $a$ violates $\chi_c$, which is exactly when $A_c$ attacks $\mathit{arg}(a)$; then $\mathit{arg}(a)\notin\mathit{St}$ by policy priority and $\Gdep\mathit{standing}$ fails at it, or $\mathit{arg}(a)$ is undefined and $\mathit{grounded}(a)$ is false outright. For the converse, under its hypotheses $\Gdep\mathit{standing}$ can fail only at $\mathit{arg}(a)$ itself, and an otherwise active argument leaves $\mathit{St}$ only through a standing attacker, which by hypothesis is $A_c$; so $\rho_c$ holds and $a$ violates $\chi_c$. Other attackers or labels make the converse fail, which is why an execution rule on grounding refuses more than $g_c$. The rest follows from the execution rule.
\end{proof}

The gate fragment is thus the part of the agent fragment where standing is decided by a rule rather than by an observation; a policy-aware support check can drive the gate.

\subsection{The assurance ledger}
\label{sec:ledger}

Two axes must be kept apart: the \emph{claim strength} a method can deliver, and the \emph{status} a requirement reaches once evidence is in.

\paragraph{Contracts.}
An obligation is a tuple $r=(\Sig,\varphi,\App_r,\kappa)$ comprising a deployment, target formula, applicability condition, and acceptance contract. $\App_r$ says when a run exercises the obligation: $\App_c$ of~\eqref{eq:rule} for a rule, $\Eventually\,\exists a\,\mathit{checked}(a)$ for the support ideal, and every run for the outcome and warning conditions. We use three contract types. An \emph{exact} contract establishes $\varphi$ on every run of $\Sig$ under stated assumptions. A \emph{suite} contract requires $N_{\mathcal T}(\App_r)\ge k$ and $N_{\mathcal T}(\neg\varphi)=0$ on a named suite $\mathcal T$. A \emph{risk} contract requires $\Pr(\neg\varphi)\le\varepsilon$ with confidence $1-\delta$ for a named $\Dist_\Sig$. A suite contract records a test result, not compliance. An \emph{evidence record} names the deployment, target, observation interface, claim type (exact, suite, upper or lower risk bound, descriptive), value, suite or distribution, validity assumptions, and lineage. A method's \emph{capability} has the same fields without a value, so ``configured but not run'' is distinct from ``run and found nothing''.

\paragraph{Status and gaps.}
Fix a finite catalogue of target formulas and a finite acyclic graph of permitted claim transformations: exact conjunction, weakening an upper bound along event inclusion, combination by Proposition~\ref{prop:compose}, and declared equivalences such as~\eqref{eq:expired}$\Leftrightarrow$\eqref{eq:past}. An obligation is \emph{addressed} if a configured capability reaches it through the graph, and \emph{established} if an evidence derivation reaches it with matching deployment and suite or distribution, supported assumptions, and a typed claim meeting the contract; a violation, an unresolved assumption, or a descriptive statistic cannot discharge an exact contract. With $\Req^{\mathrm{est}}\subseteq\Req^{\mathrm{addr}}\subseteq\Req$, the ledger reports the \emph{expressiveness gap} $G_{\mathrm{expr}}=\Req\setminus\Req^{\mathrm{addr}}$ of obligations no configured method addresses and the \emph{establishability gap} $G_{\mathrm{est}}=\Req^{\mathrm{addr}}\setminus\Req^{\mathrm{est}}$ of addressed obligations whose contract is unmet, each entry with its reason: an observed counterexample, insufficient applicability, absent evidence, an unresolved assumption, or a scope mismatch.

\begin{proposition}[Finite conditional ledger]
\label{prop:ledger}
For a finite claim graph, decidable scope and contract checks, and supplied assumption statuses, addressedness and establishment are decidable; if the leaf claims and transformations are sound under their assumptions, every established entry is sound under the accumulated assumptions; and, for a fixed requirement set $\Req$, adding capabilities never enlarges $G_{\mathrm{expr}}$.
\end{proposition}
\begin{proof}
Process the graph in topological order, propagating capability reachability and evidence derivations separately and accepting a derived claim only through its declared transformation; finiteness gives termination and induction over the order gives conditional soundness. New leaves preserve existing capability paths.
\end{proof}

The value of the proposition is its last clause: ``add a method'' is an operation whose effect on the ledger is computable. To read the ledger, begin with an obligation and its contract, ask whether a configured method can observe enough information to address it, and then compare the evidence type, scope, and assumptions with the contract. This order prevents a result for another deployment or a descriptive statistic from being mistaken for the required claim. For example, the gate establishes the exact contract for $\varphi_7$ on $\Sig_g$, and the suite records 30 violations of $\varphi_7$ on $\Sig$; conjoining the first with suite results on $\Sig$ to claim the policy on $\Sig$ is rejected, because no transformation in the graph changes the deployment. The valid derivation keeps two entries: $\varphi_7$ exact on $\Sig_g$, unmet on $\Sig$.

\paragraph{Transferring an outcome guarantee.}
Let $V=\{\tau:\neg\varphi(\tau)\}$ be a violation event and $F$ the oracle failure event. The omitted event is $B=V\cap F^c$, the violations the oracle passes.

\begin{proposition}[The missing compliance term]
\label{prop:transfer}
If valid evidence bounds $\Pr(F)\le u$ at confidence $1-\delta_F$ and $\Pr(B)\le b$ at confidence $1-\delta_B$, then $\Pr(V)\le\min\{1,u+b\}$ with confidence at least $1-\delta_F-\delta_B$. A bound on $\Pr(F)$ alone gives no bound below one on $\Pr(V)$ without a relation between $V$ and $F$.
\end{proposition}
\begin{proof}
$V=(V\cap F)\cup B$ disjointly, so $\Pr(V)\le\Pr(F)+\Pr(B)$; a union bound over the two confidence failures gives the confidence, without independence. For the last claim take $F=\varnothing$ and $V=\Runs(\Sig)$.
\end{proof}

The object is the obligation to measure or exclude $B$; repeating the outcome evaluation, calibrating a predictor of outcome labels~\cite{bates2021rcps}, or improving that predictor does not supply it. A monitor trained on outcome labels gives no evidence about $V$.

\paragraph{Combining bounds.}
\begin{proposition}[Composition of upper bounds]
\label{prop:compose}
For obligations $\varphi_1,\ldots,\varphi_m$ on one deployment distribution with valid bounds $\Pr(\neg\varphi_j)\le\varepsilon_j$ at confidence $1-\delta_j$,
$\Pr(\neg\bigwedge_j\varphi_j)\le\min\{1,\sum_j\varepsilon_j\}$ at confidence at least $1-\sum_j\delta_j$, under the joint assumptions and with no independence required.
\end{proposition}
\begin{proof}
A union bound over evidence draws, then a union bound over violation events.
\end{proof}

An exact claim contributes zero risk under its assumptions; a suite result contributes no deployment-risk term without a generalisation argument; a lower bound cannot stand in for an upper one; a gate's zero on $\Sig_g$ cannot be added to bounds on the ungated $\Sig$ without a transfer argument; and two sources sharing labels are one witness, not two.

\section{Evaluation on telecom trajectories}
\label{sec:study}

\subsection{Setup}
\label{sec:setup}

\paragraph{Corpus and policy.}
We use the released $\tau^2$-bench telecom trajectories~\cite{tau2bench} with the default policy: 114 tasks with four trials each, 456 runs, assistant \texttt{claude-3-7-sonnet-20250219}, simulated customer \texttt{gpt-4.1-2025-04-14}, as distributed with the benchmark. Twenty when--shall rules were extracted from the policy and troubleshooting workflow with source spans, then human-audited (two repaired); they are the reference specification, not a validated complete account (Table~\ref{tab:rules}).

\paragraph{Instruments.}
The rule checker evaluates $\App_c$ and $\Viol_c$ for every rule on every run. Four rules ($c_{12},c_{17},c_{18},c_{20}$) use task fault metadata available to the assessor but not the agent. A second, separately written implementation from the same rule texts agrees on all $20\times456$ pairs; it shares the rule interpretation, so it checks the code, not the interpretation. The \emph{support stand-in} uses two deterministic predicates on the agent's prefix in place of $\mathit{ext}$, keeping interpreter error out of the numbers. A closing claim that the fault is fixed must follow the agent's own diagnostic returning healthy after its last state-changing action; a refuel must follow the agent's own reading of an exhausted allowance. The prefix monitor uses TF--IDF features over event history and the last six steps, class-balanced logistic regression, and five task-grouped folds. Its threshold is chosen on pooled out-of-fold scores for at most 10\% run-level false alarms. It is a baseline in the PrefixGuard setting, not a reproduction, and its threshold uses all labels, so monitor-only counts are descriptive. The gate replays $g_7$ at recorded restoration calls without resuming the agent. ``Rule violation'' means a violation of one of the nineteen rules other than the formatting rule $c_{11}$, violated in 452 runs. All instruments are stand-ins for the cited methods; the evaluation tests the framework's accounting, not those methods' accuracy. Rules, source spans, per-run verdicts, and scripts are available in the 1\href{https://github.com/xiaoweih/Organising-Trajectory-Evidence-for-Language-Model-Agent-Assurance}{public repository}.

\begin{table}[t]
\caption{Audited rules on the 456-run suite. Positive applicability with zero violations meets only the suite contract.}
\label{tab:rules}
\footnotesize
\begin{tabularx}{\columnwidth}{@{}lXrr@{}}
\toprule
Rule & Requirement (abbreviated) & App. & Viol.\\
\midrule
$c_1$ & Identify customer before account write & 302 & 0\\
$c_2$ & Read bill before payment request & 84 & 0\\
$c_3$ & Re-check bill after payment & 83 & 2\\
$c_4$ & Refuel at most 2 GB & 181 & 0\\
$c_5$ & Obtain plan price before refuelling & 181 & 4\\
$c_6$ & Read line record before restoration & 83 & 0\\
$c_7$ & No restoration after ended-contract record & 32 & 30\\
$c_8$ & Mandated sentence after hand-off & 269 & 0\\
$c_9$ & Attempt resolution before hand-off & 269 & 2\\
$c_{10}$ & At most one call per message & 456 & 0\\
$c_{11}$ & No prose and tool call in one message & 456 & 452\\
$c_{12}$ & Status-bar check before device fix (no-service) & 113 & 92\\
$c_{13}$ & Check network before SIM intervention & 47 & 18\\
$c_{14}$ & Confirm SIM after reseating & 218 & 118\\
$c_{15}$ & Diagnostic after device fix & 435 & 23\\
$c_{16}$ & Check usage before refuelling & 181 & 0\\
$c_{17}$ & Check app permissions for relevant fault & 140 & 37\\
$c_{18}$ & Check data restriction for relevant fault & 76 & 46\\
$c_{19}$ & Enable roaming when required & 93 & 36\\
$c_{20}$ & Network-mode check for relevant fault & 224 & 132\\
\bottomrule
\end{tabularx}
\end{table}

\subsection{RQ1: adding methods one at a time}
\label{sec:each}

We add the methods in the order a deployer would meet them; the order is illustrative, and the order-free unique counts are in RQ2. At each step we report what the method finds, then a recorded run it misses and the fragment that run's failure lies in, which motivates the next method. Step numbers are positions in the recording.

\paragraph{Oracle: 231 runs.}
The oracle fails 231 runs, some of which nothing else marks: in \texttt{8f7cc8c9} the agent works for 49 steps, breaks no audited rule, makes no claim, and never fixes the fault. It misses \texttt{b183c7f2}, scored 1.0:
\begin{quote}\small
[9--10] look up line L1002 $\to$ \emph{Suspended, contract\_end\_date 2025-01-31}\\
{}[14--27] overdue bill paid\quad[28] \texttt{resume\_line(L1002)} $\to$ \emph{resumed}\\
{}[37] \texttt{transfer\_to\_human\_agents} $\to$ \emph{Transfer successful}
\end{quote}
The oracle requires only the hand-over, which happened; the restoration violates $\varphi_7$, a formula of $\Lobs$ over the whole run.

\paragraph{Rule checker: 391 runs.}
The checker reports $\Viol_7$ for \texttt{b183c7f2} at step 28 and for 29 other runs, 22 of which the oracle rewards; across the nineteen substantive rules it flags 338 runs with 540 violated rule--run pairs (992 with $c_{11}$), and the union rises to 391. Every rule is exercised, eighteen at least fifty times; seven meet the suite contract. It misses \texttt{418e4cc9}, score 0, which breaks no rule:
\begin{quote}\small
[19--21] usage L1002 $\to$ \emph{8.7 of 15 GB}; ``you haven't exceeded your data limit''\\
{}[28--42] network mode, VPN, roaming, reboot $\to$ \emph{No Connection}\\
{}[45--47] usage again $\to$ \emph{8.7 of 15 GB}; ``you still have plenty of data''\\
{}[48] \texttt{refuel\_data(L1002, 1 GB)} $\to$ \emph{\$2.00 charged}
\end{quote}
The top-up is under the 2\,GB limit and its price was checked, so $c_4$, $c_5$ and $c_{16}$ hold. The agent's facts are correct and current, and it states them; what is wrong is the decision, which rests on a premise the agent itself found false. That is an inner-tier formula in $\Lloc$.

\paragraph{Support stand-in: 401 runs.}
The support stand-in flags the refuel and 116 other runs: 112 with an unsupported closing claim (123 of 181 claims) and 5 with an unsupported refuel; 71 of the 117 pass the oracle, a second failure class the oracle rewards. A grounded action can still violate a rule (\texttt{ae5024a8} restores an expired line with every check in hand), and an ungrounded one can turn out right (\texttt{b662e94d}). The union rises to 401. It misses \texttt{0c00aad5}, an MMS fault with six causes, where over 107 steps no action lacks support and the run still fails: the failure is a drift, not a step, and every method so far decides a formula only after the relevant events.

\paragraph{Prefix-monitor baseline: 407 runs.}
The outcome-trained monitor's score crosses its threshold at step 4 of \texttt{0c00aad5}; six runs the oracle passes are marked by the monitor alone, taking the union to 407. Its target is its labels' event: trained on outcomes it detects 161 of 231 oracle failures and 0 of 30 $c_7$ breaches; trained on $c_7$ it detects 22 of 30, with 42 of its 64 alarms on non-violating runs. A second released configuration (assistant \texttt{gpt-4.1-2025-04-14}, workflow policy) shows the same contrast: 0, 22 and 20 of its 22 breaches for outcome-, rule- and $c_7$-trained monitors. As Proposition~\ref{prop:transfer} predicts, an outcome-trained monitor inherits the oracle's blind spot. Every method so far detects; none stops step 28 of \texttt{b183c7f2}.

\paragraph{Gate: 407 runs, one exact claim.}
Replay of $g_7$ rejects all 30 recorded restoration calls; by Proposition~\ref{prop:subsume}, entering $c_7$ as a standing argument renders all thirty ungrounded. The gate adds no run to the union but changes one obligation's claim type from a suite count to an exact invariant on $\Sig_g$; steps 30--37 of \texttt{b183c7f2} are not a run of $\Sig_g$.

\subsection{RQ2: do the methods differ?}
\label{sec:merged}

We compare four run-level verdicts: rule violation ($C$), support failure ($G$), outcome-trained alarm ($M$), and oracle failure ($O$). Table~\ref{tab:overlap} also gives the joint count expected if two sources with the observed flag rates marked runs independently. These expected counts are descriptive references, not significance tests; they treat runs as independent, although the four trials of a task are clustered. The rule checker's observed overlaps are close to those references: 102 versus 87 with support, 131 versus 136 with the monitor, and 178 versus 171 with the oracle. Support overlaps less with the oracle and monitor than the reference counts suggest; 71 of its 117 flagged runs pass the oracle. By contrast, the monitor closely repeats the oracle because it was trained on oracle labels: their disagreement is 0.20, compared with an independence reference of 0.50. Each source also contributes flags that no other source makes: 95 runs are flagged only by the rule checker, 5 only by support, 6 only by the monitor, and 7 only by the oracle. Without the rule checker, the other three reach 312 runs. The 22 rewarded restorations are reported by the rule checker alone. The verdicts are therefore neither interchangeable nor independent, and their counts should not simply be added.

\begin{table}[t]
\caption{Pairwise overlap of run-level verdicts. ``Exp.'' is the product-of-marginals reference under run-level independence; ``Dis.'' is the fraction of runs on which the verdicts differ, with its independence reference.}
\label{tab:overlap}
\small
\begin{tabular}{@{}lrrrrrr@{}}
\toprule
Pair & Both & Exp. & Only 1st & Only 2nd & Dis. & Indep.\\
\midrule
$C$/$G$ & 102 & 87 & 236 & 15 & .55 & .62\\
$C$/$M$ & 131 & 136 & 207 & 52 & .57 & .55\\
$C$/$O$ & 178 & 171 & 160 & 53 & .47 & .50\\
$G$/$M$ & 30 & 47 & 87 & 153 & .53 & .45\\
$G$/$O$ & 46 & 59 & 71 & 185 & .56 & .50\\
$M$/$O$ & 161 & 93 & 22 & 70 & .20 & .50\\
\bottomrule
\end{tabular}
\end{table}

\subsection{RQ3: what the merged evidence establishes}
\label{sec:ledgerstudy}

Table~\ref{tab:ledger} is the merged ledger. Seven rule contracts are established on the suite and one invariant on the gated deployment; the same prohibition, $\varphi_7$, is unmet on $\Sig$ (30 violations) and established exactly on $\Sig_g$, so status depends on the system and the contract, not on the requirement alone. The thirteen violated rules, the support ideal, and the outcome and warning ideals are addressed but unmet; the policy risk and the true-state prohibition are unaddressed.

\begin{table*}[t]
\caption{The merged ledger, conditional on the audited rule semantics and correct implementation. ``Established on suite'' never means established on all deployment runs; the gated rows concern a different deployment.}
\label{tab:ledger}
\small
\begin{tabularx}{\textwidth}{@{}>{\raggedright\arraybackslash}p{3.6cm}p{0.8cm}>{\raggedright\arraybackslash}p{1.6cm}>{\raggedright\arraybackslash}p{3.2cm}>{\raggedright\arraybackslash}X@{}}
\toprule
Obligation & System & Contract & Status & Evidence or missing premise\\
\midrule
$c_1,c_2,c_4,c_6,c_8,c_{10},c_{16}$ & $\Sig$ & Suite, $k=1$ & Established on suite & Applicable and unviolated on the 456 runs.\\
Other 13 rules, incl.\ $c_7$ & $\Sig$ & Suite, $k=1$ & Addressed, not established & At least one observed violation each.\\
Support ideal & $\Sig$ & Suite, $k=1$ & Addressed, not established & 117 runs with an ungrounded action.\\
Outcome $\neg F$ & $\Sig$ & Suite, $k=1$ & Unmet: counterexamples & 231 failures.\\
Outcome risk $\le\varepsilon$ & $\Sig$ & Risk & Not established: scope & No deployment distribution; runs clustered by task.\\
Warning ideal & $\Sig$ & Exact & Addressed, not established & Retrospective operating points do not establish the ideal.\\
Recorded $c_7$ prohibition & $\Sig_g$ & Exact & Conditionally established & Gate invariant, routing, retained history, Eq.~\eqref{eq:past}.\\
Other 19 rules; support & $\Sig_g$ & Suite & Addressed, no evidence & No gated continuations.\\
Policy risk $\Pr(\neg\bigwedge_c\varphi_c)$ & $\Sig$ & Risk & Unaddressed & No risk-contract method configured; Prop.~\ref{prop:transfer}'s $B$ unmeasured.\\
Policy risk, gated & $\Sig_g$ & Risk & Unaddressed & Separate obligation; no evidence transfers from $\Sig$.\\
True-state $c_7$ ($e^\star$) & all & Exact & Unaddressed & World predicate; no source reads it.\\
Rules absent from $C$ & all & any & Unaddressed & Outside the audited inventory.\\
\bottomrule
\end{tabularx}
\end{table*}

\section{Placing further method families}
\label{sec:stress}

A framework for organising methods should also represent methods it was not built around. We consider six further families, chosen to cover requirements our evaluation does not (information flow, adversarial input, relations between runs) together with two that sit close to existing rows (conformance, enforcement): for each, we write its target as a formula, identify the information that decides it, and state its claim type (Table~\ref{tab:stress}). There are three outcomes.

\begin{table}[t]
\caption{Further method families placed in the framework. ``Fits'': an existing fragment; ``new map'': a new observation map or distribution; ``relational'': a property of several runs, outside the current logic.}
\label{tab:stress}
\footnotesize
\begin{tabularx}{\columnwidth}{@{}>{\raggedright\arraybackslash}p{2.5cm}>{\raggedright\arraybackslash}Xl@{}}
\toprule
Method family & Target & Placement\\
\midrule
Workflow conformance~\cite{declare4py,pm4py} & ordering and eventual obligations & fits $\Lobs$\\
Specified enforcement~\cite{agentspec,ctrace} & approval, budget, consent guards & fits $\Lact$\\
Information flow~\cite{agentarmor,toolprivacybench} & purpose-bound disclosure $\varphi_{\mathrm{flow}}$ & new map\\
Adversarial scenarios~\cite{agentdojo} & authorised actions under injection & new distribution\\
Repeated-run stability & $\Pr[q(\tau)\neq q(\tau')]\le\varepsilon$ & relational\\
Counterfactual pairs & decision tracks permission, not cues & relational\\
\bottomrule
\end{tabularx}
\end{table}

\paragraph{Placed without change.} Declarative conformance checkers~\cite{declare4py,pm4py} decide finite-trace constraints over the complete log, the fragment and formula class of the rule checker; their role is an independent implementation of $\Lobs$, discharging the assumption that the checker implements the audited formulas. Runtime-enforcement languages~\cite{agentspec,ctrace} are instances of the gate: constraints such as $\Always\,\forall a\,(\mathit{risky}(a)\wedge\Exec(a)\Rightarrow\Past\,\mathit{approved}(a))$ have the guard shape of $\Lact$, and their reported extractor noise is the faithfulness assumption we attach to $\mathit{ext}$.

\paragraph{Placed by a new observation map.} Information-flow requirements, that data an executed action carries reach only sinks the purpose permits, need provenance atoms the log does not record. Provenance tracking~\cite{agentarmor,toolprivacybench} supplies an instrumented log whose observation map defines a fragment $\mathcal L_{\mathrm{flow}}\not\subseteq\Lobs$ by Proposition~\ref{prop:info} unchanged. Adversarial testing~\cite{agentdojo} keeps the formula but changes the distribution; scope matching then credits attack-suite evidence to $\Dist_{\mathrm{adv}}$, not $\Dist_\Sig$, without a transfer argument.

\paragraph{Beyond the logic.} Stability testing compares repeated runs, and counterfactual testing compares runs differing only in irrelevant cues or in permission. Both are hyperproperties~\cite{hyperproperties}, properties of sets of runs, and no formula over a single trajectory states them. This bounds the logic itself, not the configured methods; a relational extension over run tuples would admit them.

Four of the six families are represented by an existing fragment, a new observation map, or a new distribution; the other two locate a precise extension. The exercise also gives an admission rule for the pool: add a method when it decides a new fragment, matches a new distribution, strengthens a claim, or finds requirements.

\section{The residual and how to reduce it}
\label{sec:residual}

The evaluation leaves a residual with three parts, and the framework says what each needs.

\paragraph{Unflagged runs.} In 49 runs no method found a failure. This is not a compliance claim: each method decides only its target formula, and a failure outside all targets leaves all of them silent. Only an independent audit, or a checker for an untargeted formula, reduces it.

\paragraph{Unmet and unaddressed obligations.} Counting does not close the establishability gap. Rule $c_{10}$ applies in all 456 runs with no violation. Taking the 114 tasks as the independent units, the one-sided 95\% upper bound on its violation rate is $1-0.05^{1/114}\approx2.6\%$; treating the clustered runs as independent would give 0.66\%, an optimistic figure. The bound concerns that rule's wording only: $c_{16}$, ``check usage before refuelling'', also has no violation in 181 applicable runs, yet \texttt{418e4cc9} satisfies it while refuelling against its own reading. Rare-event estimation reaches below $1/N$ by re-sampling from intermediate stages~\cite{scarce}, which a recorded corpus cannot do, and its guarantee bounds overshoot, giving a floor on failure probability where a risk contract needs a ceiling. The policy-risk obligation therefore stays unaddressed, and Proposition~\ref{prop:transfer} says why no outcome bound substitutes for it.

\paragraph{Undiscovered requirements.} The framework establishes only what is written down, and the fourth failure of the introduction shows why this matters. When true facts are selected to mislead~\cite{lyingtruths}, every fact passes, every rule over events holds, and an interpreter may accept the curated facts as observations, so the support check passes too. The failing property concerns where the evidence came from and what was omitted; no configured method decides it; it must first be written down.

\paragraph{Two fronts.} \emph{Discovery} enlarges $\Req$: composition analysis finds rules no document states~\cite{skillfuzz}, and rule extraction turns policy text into formulas~\cite{skillcov}. \emph{Establishment} moves obligations from unaddressed to established: a checker for an obligation's fragment shrinks $G_{\mathrm{expr}}$ by Proposition~\ref{prop:ledger}, and stronger evidence, such as a gate in place of a count, shrinks $G_{\mathrm{est}}$. Nothing certifies that none remain undiscovered.

\section{Related work}
\label{sec:related}

\paragraph{Assurance cases and dependent evidence.}
Safety cases attach evidence and assumptions to scoped claims~\cite{kelly1998}; AMLAS applies this to machine-learned components~\cite{amlas}; multi-legged arguments explain why diverse evidence need not be independent support~\cite{bloomfield2003,littlewood2007}; dynamic assurance cases track claims through change~\cite{denney2015,calinescu2018}. Claim--argument--evidence notations such as GSN leave the logic of claims and the scope of evidence to the analyst; our ledger adds addressedness computed from information fragments, explicit matching of deployment, suite and distribution, no evidence transfer across deployments, and an expressiveness gap that shrinks as capabilities are added for fixed $\Req$.

\paragraph{Verification and enforcement of agents.}
Interpreted systems make local information central~\cite{fagin1995}; temporal--epistemic model checking verifies explicit multi-agent models~\cite{lomuscio2017}; runtime verification checks observed executions~\cite{leucker2009rv}; shields enforce action constraints, including under partial observability~\cite{konighofer2017,carr2023}. Assume--guarantee reasoning composes verified components; we compose evidence of different strengths about one system. Our gate result is a conditional invariant for a stated predicate, not a synthesised shield.

\paragraph{Evaluation methods for agents.}
Natural-language-to-temporal-logic tools and specification mining address how requirements are obtained~\cite{nl2spec}, which our framework treats as discovery; $\tau^2$-bench evaluates interaction with tools and users~\cite{tau2bench}. Skill Coverage, SkillFuzz, Grounded Continuation, PrefixGuard and SCARCE~\cite{skillcov,skillfuzz,gc,prefixguard,scarce} each supply one kind of evidence; this paper supplies the language in which their targets are formulas and the account of what they jointly establish.

\section{Limitations}

\emph{Specification}: the twenty rules are one audited interpretation of one policy, not a validated complete account. \emph{Instruments}: the rule checker's second implementation shares its interpretation; the support and monitor instruments are stand-ins, the support predicates were audited by a model, not by people, and gate replay shows no post-refusal effect. \emph{Statistics}: four trials per task are clustered, monitor thresholds are retrospective, and we report suite counts, not operational rates. \emph{Deployment}: results concern one configuration $\Sig$; a new model, prompt or tool set is a new deployment. \emph{Generality}: one telecom collection; other domains may need atoms the logs do not record, and the further families of Section~\ref{sec:stress} are represented, not run. \emph{Interpreter}: the support check is exact for the structure $\mathit{ext}$ builds, so a faithful but incomplete extraction silently narrows what it can detect.

\section{Conclusion}

We have shown that heterogeneous evidence about language-model agents can be organised in one two-tier logic, whose fragments say which method decides each requirement and whose ledger states what the merged evidence establishes; on telecom, the methods jointly flag 407 of 456 runs yet establish only suite-scoped claims and one exact gated invariant. Future work covers relational requirements, human validation of $\mathit{ext}$, feeding discovered requirements into the ledger, and risk contracts under deployment distributions in other domains.

\clearpage
\balance
\bibliographystyle{ACM-Reference-Format}
\bibliography{references}
\end{document}